%% file: cups.tex
\documentclass[11pt,a4paper]{article}
\pdfoutput=1
\usepackage{amssymb,amsmath,amsthm,cite,color,graphicx,microtype}
\usepackage[margin=34mm]{geometry}
\usepackage{hyperref}

\hypersetup{
    colorlinks=true,
    linkcolor=black,
    citecolor=black,
    filecolor=black,
    urlcolor=black,
    pdfauthor={Michael A. Bender, S\'andor P. Fekete, Alexander Kr\"oller, Vincenzo Liberatore, Joseph S. B. Mitchell, Valentin Polishchuk, Jukka Suomela},
    pdftitle={The minimum backlog problem}
}

\newtheorem{theorem}{Theorem}[section]
\newtheorem{lemma}[theorem]{Lemma}
\newtheorem{corollary}[theorem]{Corollary}

\theoremstyle{remark}

\newcommand{\ov}{\overline}
\DeclareMathOperator{\diam}{diam}
\newcommand{\mysize}[1]{\lvert #1 \rvert}
\newcommand{\mynatural}{\mathbb{N}}
\newcommand{\myreal}{\mathbb{R}}
\newcommand{\iv}[2]{[\, {#1}, \, {#2} \,]}
\newcommand{\Biv}[2]{\big[\, {#1}, \, {#2} \,\big]}
\newcommand{\W}[3]{W([\, {#1}, \, {#2} \,], \, {#3})}
\newcommand{\BW}[3]{W \big(\big[\, {#1}, \, {#2} \,\big], \, {#3}\big)}

\begin{document}

\begin{titlepage}
\author{Michael~A.~Bender, S\'andor~P.~Fekete, \\ Alexander~Kr\"oller, Vincenzo~Liberatore, \\ Joseph~S.~B.~Mitchell, Valentin~Polishchuk, \\ Jukka~Suomela}

\title{The Minimum Backlog Problem}

\date{}

\maketitle

\begin{abstract}
We study the minimum backlog problem (MBP). This online problem arises, e.g.,\ in the context of sensor networks. We focus on two main variants of MBP.

The \emph{discrete MBP} is a 2-person game played on a graph $G=(V,E)$. The \emph{player} is initially located at a vertex of the graph. In each time step, the \emph{adversary} pours a total of one unit of water into \emph{cups} that are located on the vertices of the graph, arbitrarily distributing the water among the cups. The player then moves from her current vertex to an adjacent vertex and empties the cup at that vertex. The player's objective is to minimize the \emph{backlog}, i.e., the maximum amount of water in any cup at any time.

The \emph{geometric MBP} is a continuous-time version of the MBP: the cups are points in the two-dimensional plane, the adversary pours water continuously at a constant rate, and the player moves in the plane with unit speed. Again, the player's objective is to minimize the backlog.

We show that the \emph{competitive ratio} of any algorithm for the MBP has a lower bound of $\Omega(D)$, where $D$ is the diameter of the graph (for the discrete MBP) or the diameter of the point set (for the geometric MBP). Therefore we focus on determining a strategy for the player that guarantees a uniform upper bound on the absolute value of the backlog.

For the absolute value of the backlog there is a trivial lower bound of $\Omega(D)$, and the deamortization analysis of Dietz and Sleator gives an upper bound of $O(D\log N)$ for $N$ cups. Our main result is a tight upper bound for the geometric MBP: we show that there is a strategy for the player that guarantees a backlog of $O(D)$, independently of the number of cups.

We also study a \emph{localized} version of the discrete MBP: the adversary has a location within the graph and must act locally (filling cups) with respect to his position, just as the player acts locally (emptying cups) with respect to her position. We prove that deciding the value of this game is PSPACE-hard.
\end{abstract}
\thispagestyle{empty}
\setcounter{page}{0}
\end{titlepage}

\section{Introduction}

We study the \textbf{minimum backlog problem} (\textbf{MBP}). This is an online problem in which an agent moves around a domain and services a set of locations, ``emptying'' a buffer at each location in an effort to make sure that no buffer gets too full. The MBP is related to the $k$-server problem \cite{chrobak1991an,fiat1990competitive,fb-olksp-97,210128,manasse1990competitive,sleator1985amortized} (with $k=1$), in which requests are popping up at points in a metric space, and the $k$ servers need to minimize the distance traveled to satisfy the requests. However, in the MBP, the objective is not to minimize the traveled distance, but to minimize the \emph{backlog}, i.e., the maximum amount of data residing in any buffer at any point in time.

\subsection{Motivation}

A practical motivation for the MBP arises in the context of a sensor network that, e.g., performs motion-tracking for a set of objects that move within the sensor field. Each sensor acquires data about nearby objects. The total rate of data accumulation within the network remains approximately constant, assuming a relatively fixed set of objects being monitored; however, the distribution of the data rate over the field is nonuniform and unpredictable.

If the system is used for a field study where the data is not analyzed until the end of the experiment, it may be much more energy-efficient to store the bulk of the data locally on a memory card and let someone or something gather the data by physically visiting the sensor device (or a neighborhood of the device) \cite{diao2007rethinking,gu2006dhm,jea2005mcm,mdgps-ulpdssn-06,somasundara2004mes}. During the experiment, each sensor only needs to report the amount of data in its local buffer. The objective of the data gatherer is to visit the sensors in an effective order, so that no sensor's storage device is overfull.

An analogous problem arises in scheduling battery recharging/replacement in a field of wireless devices whose power
consumption varies unpredictably with time and location.

\subsection{Discrete MBP}

We will now formalize three versions of the MBP that we study in this paper. We start with the \textbf{discrete MBP}. In this problem, we have an unweighted directed graph $G=(V,E)$ with an (initially empty) \textbf{cup} on each vertex. There is a \textbf{player}, who moves from vertex to vertex along the edges of the graph emptying cups, and there is an \textbf{adversary} (not located anywhere in particular), who refills the cups with water. We talk about filling cups with water because of historical precedent~\cite{DietzSl87}.

The following game is played for an indefinite (but finite) number of rounds. In each round, the two opponents do the following:
\begin{itemize}
    \item The adversary pours a total of one unit of water into the cups. The adversary is free to distribute the unit of water in any way he likes. The adversary may base his decision on the current location of the player and the current water levels in all of the cups.
    \item The player moves along an edge and empties the cup in its new location. The player can see the amount of water that the adversary has poured into each cup, and the player can use this information to make the decisions.
\end{itemize}
The problem is online, i.e., the player does not know how the water is distributed in the future. Thus, the player's decision of which edge to traverse may be based only on the amount of the water that has been poured into all cups so far, but not on the future distribution of water. The objective of the player is to minimize the \textbf{backlog}, which is defined to be the maximum amount of water in any cup at any time.

\subsection{Geometric MBP}

It is straightforward to generalize the MBP to weighted, continuous scenarios. In this paper we will mainly focus on the following version which we call the \textbf{geometric MBP}; this version is of particular interest for sensor-network applications. Let $P \subseteq \myreal^2$ be a finite planar set. There is a cup on each point of~$P$. The two-player game proceeds as follows in a continuous manner:
\begin{itemize}
    \item The adversary pours water into the cups $P$; the total rate at which the water is poured into the cups is~1.
    \item The player moves in the plane with unit speed, starting at an arbitrary point. Whenever the player visits a cup, the cup is emptied.
\end{itemize}
Again, this is an online problem, and the goal of the player is to minimize the backlog.

\subsection{Localized MBP}\label{sec:localizedMBP}

In the discrete and geometric versions of the MBP, the actions of the player are restricted by her location. To keep the game fair, we may also consider the following variant in which the actions of the adversary are also restricted by his location; we call it the \textbf{localized MBP}. The game is a fairly straightforward adaptation of the discrete MBP; however, some of the details need more care, as we will be interested in deciding the exact value of this game.

The localized MBP is a two-player game played on a directed graph $G = (V,E)$. Each vertex is a cup that carries some integer load (water level). An instance of the localized MBP consists of the graph, the initial loads of the cups, and the distinct starting positions of the two players.

At each time step, both the player and the adversary are located at some vertices. The player starts the game, and both participants take turns in moving from their respective current position along an outgoing edge to an adjacent node (they are not allowed to stay in place):
\begin{itemize}
    \item A move by the player ends with her removing the load from the vertex she reached.
    \item A move by the adversary ends with him increasing the load on the vertex he reached by one unit.
\end{itemize}
The game ends when the player steps onto the vertex currently occupied by the adversary, or when the adversary manages to get the load on some vertex to a pre-specified \textbf{target value}. The player wins if she can keep the adversary from reaching the target value; the adversary wins if he can reach the target value.

\subsection{Results}

Throughout this work, we write $N$ for the number of cups, and $D$ for the maximum distance between the cups (i.e., the diameter of the graph $G$ or the set $P$). We start with the following simple observations (Section \ref{sec:prelim}):
\begin{itemize}
    \item \emph{Discrete and geometric MBP, competitive analysis}: The competitive ratio of any algorithm for the problem may be as bad as $\Omega(D)$. Therefore we will focus on \emph{absolute} bounds on the backlog in this paper.
    \item \emph{Discrete and geometric MBP, lower bounds}: The adversary can guarantee a backlog of $\Omega(D)$. For the discrete version, the adversary can also guarantee a backlog of $\Omega(\log N)$.
    \item \emph{Discrete and geometric MBP, upper bounds}: The player can guarantee a backlog of $O(D\log N)$.
\end{itemize}
Our main results are as follows:
\begin{itemize}
    \item \emph{Geometric MBP, upper bounds} (Sections \ref{sec:geom-strategy} and \ref{sec:geom-analysis}): We show that the player can achieve a backlog of $O(D)$, independently of the number of cups. This is optimal up to constants.
    \item \emph{Localized MBP, hardness} (Section~\ref{sec:local-hard}): We show that deciding the value of the game is PSPACE-hard, even for the smallest nontrivial target value of~$2$.
\end{itemize}

\section{Background and Related Work}

This style of problem, with a player emptying one cup at a time and an adversary distributing water among cups, is a classic problem, which has been independently discovered and rediscovered many times. However, in previous formulations, there is no issue of locality---that is, the player can empty any single cup (or, depending on the formulation, any feasible subset of cups) in any time step, independently of her previous actions.

The earliest reference to cup emptying of which we are aware is the work of Dietz and Sleator~\cite{DietzSl87}, who used it as a technique for deamortizing data structures. They proved that if there are $N$ cups, and the player always empties the fullest one, then no cup ever contains more than $\ln N$ water; they also showed that this bound is optimal. The bound leads to an optimal, worst-case-constant algorithm for the order-maintenance problem. This deamortization technique has been used many times since.

Adler et al.~\cite{AdlerBeFrGoGoPa03} used cup emptying as a technique for analyzing scheduling algorithms. In particular, they showed how to use cup emptying to produce ``fair'' job schedules. Chrobak et al.~\cite{ChrobakCsImNoSsWo01} introduced a generalization of cup emptying and applied it to multiprocessor scheduling with conflicts between tasks; there is subsequent work by, e.g., Bar-Noy et al.~\cite{bfln-cosp-02}, Koga~\cite{k-bstlfpqdf-01}, and Rote~\cite{Rote03}. In the problem, the player can empty more than one cup at a time, but there is a conflict graph between cups. If two cups conflict, only one of them can be emptied at a time. Here there is a graph, but there is still no issue of locality.

\pagebreak

More recent work by Bodlaender et al.~\cite{bhw-cgha-11} considers a generalization of the game on an unweighted, undirected graph $G=(V,E)$, in which one move by the adversary consists of arbitrarily distributing one unit of water among a subset $S\subseteq V$, while a move of the player is to empty all cups of an independent set. They show that the value of the game (the largest possible filling level of a cup that the adversary can force, and to which the player can limit the sequence) lies between the natural logarithm of the clique number and the natural logarithm of the chromatic number, settling the value of the game for all perfect graphs. They also consider the game on the simplest non-perfect graphs, i.e., odd holes and odd anti-holes, and show that a natural greedy strategy is generally not optimal. Bodlaender et al.~\cite{bhk+-cvws-12} consider another variant on a ring graph with $N$ nodes, where the player may empty an arbitrary group of $c$ consecutive cups in each round, and compute the exact values for all $N\leq 12$.

To the best of our knowledge, in all previous work on cup emptying, there is no concept of locality of the player with respect to the cups. In contrast, the minimum backlog problem studied in this paper is specifically cup emptying with a player who moves around in graphs or geometric domains.

It is important to note that although the motivation for studying cup emptying in metric spaces came from online vehicle routing, the locality shows up in non-geometric contexts as well. For instance, in a job scheduling problem, there may be a (set-up) cost associated with switching from executing one job to executing another. This, in particular, makes the Traveling Salesman Problem, which originated as a geometric problem, applicable also to scheduling tasks \cite{review,flowshop,separation}. Hence, although we state our problem and results in purely geometric terms, they are also relevant in some more general scheduling applications.

\section{Preliminary Observations}\label{sec:prelim}

We start with preliminary observations on the discrete and geometric versions of the MBP.

\subsection{Competitive Analysis is Doomed}

It would be natural to try to give a competitive algorithm for the MBP. Unfortunately, an online algorithm with a good competitive ratio is not possible, unless we restrict ourselves to the case of $N = O(1)$ or $D = O(1)$. Indeed, consider either of the following scenarios, both of which have $N$ cups and a diameter of $D = N-1$:
\begin{enumerate}
    \item \emph{discrete MBP}: graph $G$ is a path on $N$ vertices,
    \item \emph{geometric MBP}: set $P$ consists of the points $(1,0),\ (2,0),\ \dotsc,\ (N,0)$.
\end{enumerate}
Number the cups by $1,2,\dotsc,N$ so that cups $1$ and $N$ are the endpoints of the diameter. Suppose that the adversary picks a random permutation of the cups $2$ through $N-1$, and for the first $N-2$ time steps pours a unit of water per step into the cups according to the permutation. Then, the adversary picks one of the endpoints of the path and pours the water there forever. The best offline strategy would be to rush to the endpoint and stay there---this yields a maximum backlog of~1. On the other hand, without prior knowledge of the ``drenched'' endpoint, any algorithm will put the player far from the endpoint (since the adversary may choose the endpoint that is farthest from the current player position), thus, making the performance of the algorithm $\Omega(N) = \Omega(D)$.

Given that the competitive ratio of an online algorithm for the problem may be very high, we concentrate on providing uniform upper bounds on the performance, i.e., on giving universal bounds on the amount of water in any cup at any time.

\subsection{Simple Lower Bounds}\label{ssec:simple-lower}

The performance of any algorithm has a lower bound of $\Omega(D)$: the adversary can simply pour water at the rate of $1/2$ into the cups at the ends of the diameter. The same holds for both the discrete MBP and the geometric MBP.

To get another lower bound, the adversary may pick a set of $K$ cups such that the distance between any two cups in the set is at least $d$, for some number~$d$. The adversary will then pour the water evenly into never-emptied cups from the set. After $dK$ steps, one of the cups will have
\[
    \Omega\bigl(\tfrac{d}{K}+\tfrac{d}{K-1}+\dots+d\bigr)=\Omega(d\ln K)
\]
units of water.

In particular, for the discrete MBP with $N$ cups, we can always choose $d = 1$ and $K = N$, and hence the adversary can ensure a backlog of $\Omega(\max(D,\log N))$. There are also some families of graphs in which the adversary can guarantee a backlog of $\Omega(D \log N)$: consider, for example, subdivisions of star graphs.

\subsection{Simple Upper Bounds}\label{ssec:simple-upper}

As for the upper bounds, which is the main focus of the paper, a na\"ive algorithm has a performance of $O(C)$, where $C$ is the length of the shortest closed path visiting all cups.

To find a better upper bound, let us first consider the \emph{discrete MBP on a complete graph} (or, put otherwise, a game in which the player can empty any single cup in each round). Intuitively, always emptying the fullest cup is optimal, and a simple exchange argument validates this intuition. Dietz and Sleator~\cite{DietzSl87} analyze the performance of this strategy.
\begin{lemma}[{Dietz and Sleator~\cite[Theorem~5]{DietzSl87}}]\label{lem:cup}
    In the discrete MBP on a complete graph with $N$ vertices, if the player always empties the fullest cup, then no cup ever contains more than $\ln{N}$ units of water.
\end{lemma}

We can apply the same strategy in the discrete and geometric MBP: the player repeatedly walks to the fullest cup (which takes at most $D$ time units) and empties it. A simple application of Lemma~\ref{lem:cup} shows that the backlog of this strategy will be bounded by $O(D \log N)$, where $N$ is the number of cups.

This is the upper bound that we set out to beat in this paper. Recall from Section~\ref{ssec:simple-lower} that we cannot necessarily do any better in the discrete MBP. However, we show that in the geometric MBP the player always has a strategy that guarantees a backlog of $O(D)$, independently of the number of cups.

\section{Algorithmic Techniques}

We will now start to develop the algorithmic techniques that we will use to design our strategy for the geometric MBP. There are two main ingredients:
\begin{enumerate}
    \item the $(\tau,k)$-game (Section~\ref{ssec:taukgame}),
    \item Few's lemma (Section~\ref{ssec:few}).
\end{enumerate}
The $(\tau,k)$-game is a purely combinatorial two-player game; we do not make any references to the geometric setting that we have in the geometric MBP. On the other hand, Few's lemma is a purely geometric result; there is no game-theoretic element in the statement of the result. Section~\ref{sec:geom-strategy} shows how to put the two ingredients together in order to design a strategy for the geometric MBP.

\subsection{\boldmath The \texorpdfstring{$(\tau,k)$}{(tau,k)}-game}\label{ssec:taukgame}

The $(\tau,k)$-game is a straightforward generalisation of the empty-the-fullest strategy on a complete graph (recall Lemma~\ref{lem:cup}). For each $\tau\in\myreal$, $k\in\mynatural$ we define the \emph{$(\tau, k)$-game} as follows. There is a set of cups, initially empty, not located in any particular metric space. At each time step the following takes place, in this order:
\begin{enumerate}
    \item The adversary pours a total of $\tau$ units of water into the cups. The adversary is free to distribute the water in any way he likes.
    \item The player empties $k$ \emph{fullest} cups.
\end{enumerate}

The game is discrete---the player and the adversary take turns making moves during discrete time steps. The following lemma bounds the amount of water in any cup after $r$ steps of the game; this is a direct extension of a result of Dietz and Sleator~\cite{DietzSl87}.
\begin{lemma}\label{lem:tauk-k-game}
    The water level in any cup after $r$ complete time steps of the $(\tau, k)$-game is at most $H_r \tau / k$, where $H_r$ is the $r$th harmonic number.
\end{lemma}

\begin{proof}
    We follow the analysis of Dietz and Sleator~\cite[Theorem~5]{DietzSl87}. Consider the water levels in the cups after time step $j$. Let \smash{$X_j^{}$\raisebox{4.6pt}{$\scriptstyle \mspace{-7mu}(i)$}} be the amount of water in the cup that is $i$th fullest, and let
    \[
        S_j \,= \!\!\!\!\sum_{i=1}^{(r-j)k+1}\!\!\!\! X_j^{(i)}
    \]
    be the total amount of water in $(r-j)k+1$ fullest cups at that time. Initially, $j = 0$, $\smash{X_0^{(i)} = 0}$, and $S_0 = 0$.

    Let us consider what happens during time step $j \in \{1,2,\dotsc,r\}$. The adversary pours $\tau$ units of water; the total amount of water in $(r-(j-1))k+1$ fullest cups is therefore at most $S_{j-1} + \tau$ after the adversary's move and before the player's move (the worst case being that the adversary pours all water into $(r-(j-1))k+1$ fullest cups).

    Then the player empties $k$ cups. These $k$ fullest cups contained at least a fraction $k/((r-(j-1))k+1)$ of all water in the $(r-(j-1))k+1$ fullest cups; the remaining $(r-j)k+1$ cups are now the fullest. We obtain the inequality
    \[
        S_j \le \left( 1 - \frac{k}{(r-(j-1))k+1} \right) (\tau + S_{j-1})
    \]
    or
    \[
        \frac{S_j}{(r-j)k+1} \le \frac{\tau}{(r-(j-1))k+1} + \frac{S_{j-1}}{(r-(j-1))k+1} \, .
    \]
    Therefore the fullest cup after time step $r$ has the water level at most
    \[
        \begin{split}
        X_r^{(1)} &\,=\, \frac{S_r}{k(r-r)+1} \\
            &\,\le\, \frac{\tau}{1k+1} + \frac{\tau}{2k+1} + \dotso + \frac{\tau}{rk+1} + \frac{S_{0}}{rk+1} \\
            &\,\le\, \frac{\tau}{k} \left( \frac{1}{1} + \frac{1}{2} + \dotso + \frac{1}{r} \right) .
            \qedhere
        \end{split}
    \]
\end{proof}

\subsection{Few's Lemma}\label{ssec:few}

Let us next introduce the geometric ingredient that we will need. The following result is by Few~\cite{few55shortest}:
\begin{lemma}[{Few~\cite[Theorem~1]{few55shortest}}]\label{lem:few}
    Given $n$ points in a unit square, there is a path through the $n$ points of length not exceeding $\sqrt{2n}+1.75$.
\end{lemma}

We make use of the following corollary:
\begin{corollary}\label{cor:few}
    Let $S$ be a $D \times D$ square. Let $i \in \{0,1,\dotsc\}$. Let $Q \subseteq S$ be a planar point set with $\mysize{Q}=25^i$ and $\diam(Q)=D$. For any point $p \in S$ there exists a closed tour of length at most $5^{i+1}D$ that starts at $p$, visits all points in $Q$, and returns to~$p$.
\end{corollary}

\begin{proof}
    If $i > 0$, by Lemma~\ref{lem:few}, there is tour of length at most
    \[
        \big(\sqrt{ 2 (25^i + 1) }  + 1.75 +\sqrt2 \big) D \,\le\, 5^{i+1}D
    \]
    that starts at $p$, visits all points in $Q$, and returns to $p$. If $i=0$, there is a tour of length $2 \sqrt{2} D \le 5 D$ through $p$ and $\mysize{Q}=1$ points.
\end{proof}

\section{Geometric MBP: Player's Strategy}\label{sec:geom-strategy}

Now we are ready to present an asymptotically optimal algorithm for the geometric MBP. The player's strategy is composed of a number of coroutines, which we label with $i \in \{ 0, 1, \dots \}$. The coroutine $i$ is invoked at times ${(10 L + \ell)} \tau_i$ for each $L \in \{0,1, \dotsc \}$ and $\ell \in \{1,2,\dotsc,10\}$. Whenever a lower-numbered coroutine is invoked, higher-numbered coroutines are suspended until the lower-numbered coroutine returns.

\subsection{Parameters}\label{sec:paremeters}

We choose the values $\tau_i$ as follows. For $i \in \{0, 1, 2, \dotsc \}$, let
\begin{align*}
    k_i &= 25^i, \\
    \tau_i &= (2/5)^i \cdot 10 D k_i = 10^i \cdot 10 D.
\end{align*}

For $L \in \{0,1, \dotsc \}$ and $\ell \in \{1,2,\dotsc,10\}$, define \emph{$(i,L,\ell)$-water} to be the water that was poured during the time interval $\iv{10 L \tau_i}{(10 L + \ell) \tau_i}$.

\subsection{\boldmath Coroutine \texorpdfstring{$i$}{i}}

The coroutine $i$ performs the following tasks when invoked at time ${(10 L + \ell)} \tau_i$:
\begin{enumerate}
    \item \label{step:choose-cups} Determine which $k_i$ cups to empty. The coroutine chooses to empty $k_i$ cups with the largest amount of $(i,L,\ell)$-water.
    \item \label{step:choose-cycle} Choose a cycle of length at most $\tau_i / 2^{i+1}$ which visits the $k_i$ cups and returns back to the original position. This is possible due to Corollary~\ref{cor:few} because ${5^{i+1} D = \tau_i / 2^{i+1}}$.
    \item \label{step:control-player} Guide the player through the chosen cycle.
    \item Return.
\end{enumerate}
Observe that when a coroutine returns, the player is back in the location from which the cycle started. Therefore invocations of lower-numbered coroutines do not interfere with any higher-number coroutines which are currently suspended; they just delay the completion of the higher-numbered coroutines.

The completion is not delayed for too long. Indeed, consider a time period $\iv{j \tau_i}{{(j+1)} \tau_i}$ between consecutive invocations of the coroutine~$i$. For $h \in \{0 , 1, \dotsc, i \}$, the coroutine $h$ is invoked $10^{i-h}$ times during the period (recall the definition of $\tau_i$ in Section~\ref{sec:paremeters}). The cycles of the coroutine $h$ have total length at most
\[
    10^{i-h} \tau_{h} / 2^{h+1} \,=\, \tau_i / 2^{h+1} .
\]
In grand total, all coroutines $0,1, \dotsc, i$ invoked during the time period take time
\[
     \sum_{h=0}^i \tau_i / 2^{h+1} \,<\, \tau_i .
\]
Therefore all coroutines invoked during the time period are able to complete within it. This proves that the execution of the coroutines can be scheduled as described.

\section{Geometric MBP: Analysis}\label{sec:geom-analysis}

We now analyse the backlog under the strategy of Section~\ref{sec:geom-strategy}. For any points in time $0 \le t_1 \le t_2 \le t_3$, we write $\W{t_1}{t_2}{t_3}$ for the maximum per-cup amount of water that was poured during the time interval $\iv{t_1}{t_2}$ and is still in the cups at time $t_3$.

We need to show that $\W{0}{t}{t}$ is bounded by a constant that does not depend on $t$. To this end, we first bound the amount of $(i,L,\ell)$-water present in the cups at the time ${(10 L + \ell +1)} \tau_i$. Then we bound the maximum per-cup amount of water at an arbitrary moment of time by decomposing the water into $(i,L,\ell)$-waters and a small remainder.

We make use of the following simple properties of $\W{\cdot}{\cdot}{\cdot}$. Consider any four points in time $0 \le t_1 \le t_2 \le t_3 \le t_4$. First, the backlog for \emph{old} water is nonincreasing: $\W{t_1}{t_2}{t_4} \le \W{t_1}{t_2}{t_3}$. Second, we can decompose the backlog into smaller parts: $\W{t_1}{t_3}{t_4} \le \W{t_1}{t_2}{t_4} + \W{t_2}{t_3}{t_4}$.

\subsection{\boldmath \texorpdfstring{$(i,L,\ell)$}{(i,L,l)}-Water at Time \texorpdfstring{${(10 L + \ell + 1)} \tau_i$}{(10 L + l + 1) tau\_i}}

Let $i \in \{0,1,\dotsc\}$, $L \in \{0,1,\dotsc\}$, and $\ell \in \{1,2,\dotsc,10\}$. Consider $(i,L,\ell)$-water and the activities of the coroutine $i$ when it was invoked at the times
${(10 L + 1)} \tau_i ,\allowbreak\,
{(10 L + 2)} \tau_i ,\allowbreak\dotsc,\allowbreak\,
{(10 L + \ell)} \tau_i$.

The crucial observation is the following. By the time ${(10 L + \ell + 1)} \tau_i$, the coroutine $i$ has, in essence, played a $(\tau_i, k_i)$-game for $\ell$ rounds with $(i,L,\ell)$-water. The difference is that the coroutine cannot empty the cups immediately after the adversary's move; instead, the cup-emptying takes place during the adversary's next move. Temporarily, some cups may be fuller than in the $(\tau_i, k_i)$-game. However, once we wait for $\tau_i$ time units to let the coroutine $i$ complete its clean-up tour, the maximum level of the water that has arrived before the beginning of the clean-up tour is at most that in the $(\tau_i, k_i)$-game.

The fact that the emptying of the cups is delayed can only hurt the adversary, as during the clean-up tour the player may also accidentally undo some of the cup-filling that the adversary has performed on his turn. The same applies to the intervening lower-numbered coroutines.

Therefore, by Lemma~\ref{lem:tauk-k-game}, we have
\begin{equation}\label{eq_W}
    \BW{10 L \tau_i}{(10 L + \ell) \tau_i}{(10 L + \ell +1) \tau_i}
    \,\le\, H_\ell \cdot \tau_i/k_i
    \,<\, 3 \tau_i/k_i ,
\end{equation}
using the fact that $H_\ell < 3$ for every $\ell\leq 10$.

\subsection{\boldmath Decomposing Arbitrary Time \texorpdfstring{$t$}{t}}

Let $t$ be an arbitrary instant of time. We can write $t$ as $t = T \tau_0 + \epsilon$ for a nonnegative integer $T$ and some remainder $0 \le \epsilon < \tau_0$. Furthermore, we can represent the integer $T$ as
\[
    T \,=\, \ell_0 + 10 \ell_1 + \dotsb + 10^N \ell_N
\]
for some integers $N \in \{0,1,\dotsc\}$ and $\ell_i \in \{1,2,\dotsc,10\}$. Since the range is $1 \le \ell_i \le 10$, not $0 \le \ell_i \le 9$, this is not quite the usual decimal representation; we chose this range for $\ell_i$ to make sure that $\ell_i$ is never equal to~0.

We also need partial sums
\[
    L_i \,=\, \ell_{i+1} + 10 \ell_{i+2} + \dotsb + 10^{N-i-1} \ell_N .
\]
Put otherwise, for each $i \in \{0,1,\dotsc,N\}$ we have
\[
    T \,=\, \ell_0 + 10 \ell_1 + \dotsb + 10^i \ell_i + 10^{i+1} L_i
\]
and therefore
\[
    \begin{split}
        t &\,=\, \epsilon + \ell_0 \tau_0 + \ell_1 \tau_1 + \dotsb + \ell_N \tau_N \\
        &\,=\, \epsilon + \ell_0 \tau_0 + \ell_1 \tau_1 + \dotsb + \ell_i \tau_i + 10 L_i \tau_i
    \end{split}
\]
(recall from Section~\ref{sec:paremeters} that $10^i \tau_0 = 10^i \cdot 10D = \tau_i$).

We partition the time from $0$ to $t-\epsilon$ into long and short periods. The \emph{long period} $i \in \{0,1,\dotsc,N\}$ is of the form
$\Biv{10L_i \tau_i}{{(10L_i + \ell_i-1)} \tau_i}$
and the \emph{short period} $i$ is of the form
$\Biv{{(10L_i + \ell_i-1)} \tau_i}{{(10L_i + \ell_i)} \tau_i}$.
See Figure~\ref{fig:decomp} for an illustration. Short periods are always nonempty, but the long period $i$ is empty if $\ell_i = 1$.

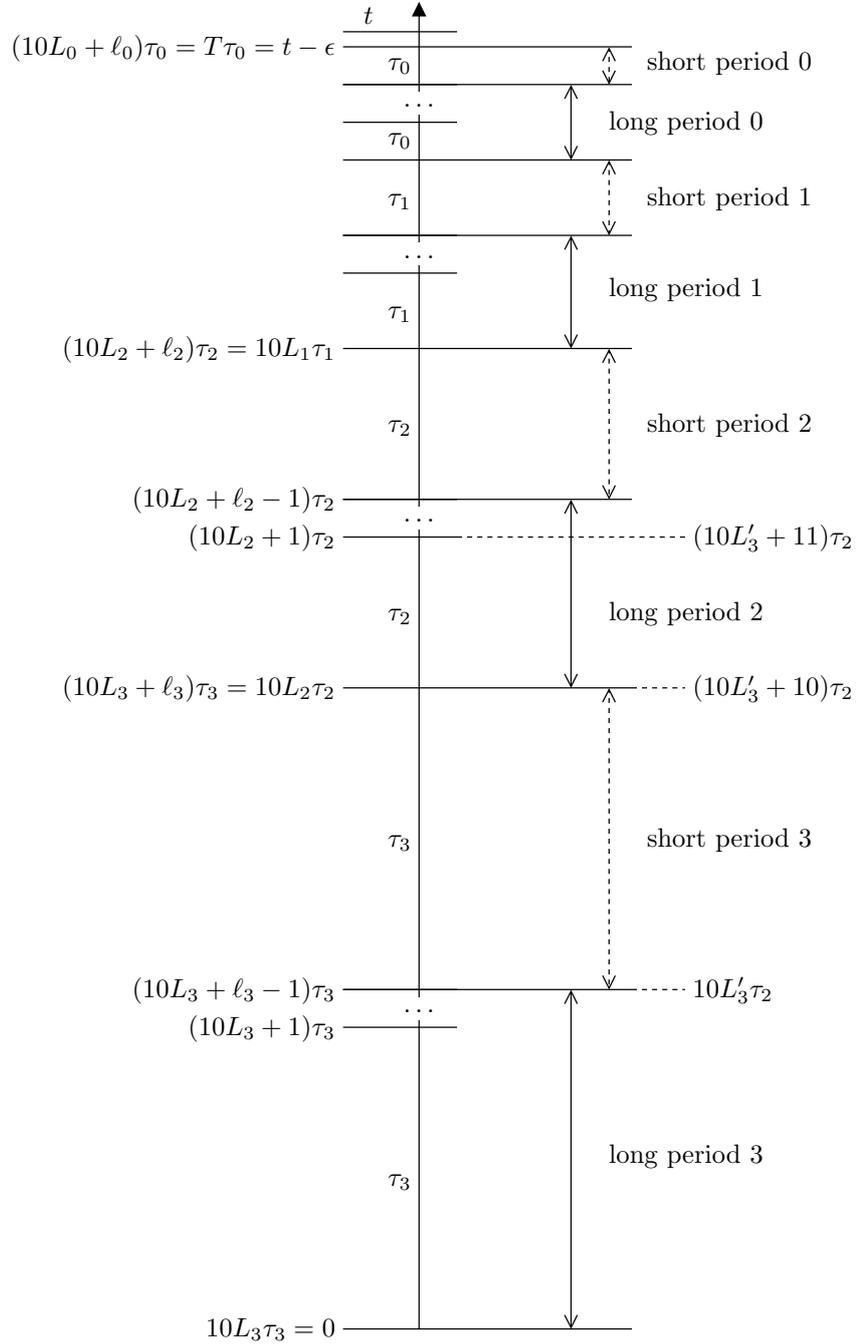
\begin{figure}
    \centering
    \input{decomp.pdftex_t}
    \caption{Decomposition of the time; in this example, $N = 3$. The illustration is not in scale; actually $\tau_i = 10 \tau_{i-1}$.}\label{fig:decomp}
\end{figure}

\subsection{\boldmath Any Water at Arbitrary Time \texorpdfstring{$t$}{t}}

Now we make use of the decomposition of an arbitrary time interval $\iv{0}{t}$ defined in the previous section: we have long periods $i \in \{0,1,\dotsc,N\}$, short periods $i \in \{0,1,\dotsc,N\}$, and the remainder $\iv{t-\epsilon}{t}$.

Consider the long period $i$. We bound the backlog from this period by considering the point in time $(10 L_i + \ell_i) \tau_i \le t$. If the period is nonempty, that is, $\ell_i > 1$, then we have by~\eqref{eq_W}
\begin{equation}\label{eq_long}
    \begin{split}
        &\BW{10 L_i \tau_i}{(10 L_i + \ell_i - 1) \tau_i}{t} \\
        &\,\le\, \BW{10 L_i \tau_i}{(10 L_i + \ell_i - 1) \tau_i}{(10 L_i + \ell_i) \tau_i} \,<\, 3 \tau_i/k_i ,
    \end{split}
\end{equation}
using the fact that
\[
    \begin{split}
        &\BW{10 L_i \tau_i}{(10 L_i + \ell_i - 1) \tau_i}{(10 L_i + \ell_i) \tau_i} \\
        &\,\leq\, \BW{10 L_i \tau_i}{(10 L_i + \ell_i) \tau_i}{(10 L_i + \ell_i + 1) \tau_i}.
    \end{split}
\]
Naturally, if the period is empty, then \eqref{eq_long} holds as well.

Consider a short period $i$ for $i > 0$. It will be more convenient to write the short period in the form
$\iv{10 L_i' \tau_{i-1}}{{(10 L_i' + 10)} \tau_{i-1}}$
where $L_i' = 10L_i + \ell_i - 1$ is a nonnegative integer (this is illustrated in Figure~\ref{fig:decomp} for $i=3$). We bound the backlog from this period by considering the point in time ${(10 L_i' + 11)} \tau_{i-1} \le t$. By~\eqref{eq_W} we have
\begin{equation}\label{eq_short}
    \begin{split}
        &\BW{{(10L_i + \ell_i-1)} \tau_i}{{(10L_i + \ell_i)} \tau_i}{t} \\
        &\,=\, \BW{10 L_i' \tau_{i-1}}{(10 L_i' + 10) \tau_{i-1}}{t} \\
        &\,\le\, \BW{10 L_i' \tau_{i-1}}{(10 L_i' + 10) \tau_{i-1}}{(10 L_i + 11) \tau_{i-1}} \,<\, 3 \tau_{i-1}/k_{i-1} .
    \end{split}
\end{equation}

Next consider the short period $i = 0$. We have the trivial bound $\tau_0$ for the water that arrived during the period. Therefore
\begin{equation}\label{eq_short0}
    \begin{split}
        \BW{{(10L_0 + \ell_0-1)} \tau_0}{{(10L_0 + \ell_0)} \tau_0}{t} \,\le\, \tau_0 .
    \end{split}
\end{equation}

Finally, we have the time segment from $t-\epsilon$ to $t$. Again, we have the trivial bound $\epsilon$ for the water that arrived during the time segment. Therefore
\begin{equation}\label{eq_rem}
    \begin{split}
        \BW{t-\epsilon}{t}{t} \,\le\, \epsilon \, < \, \tau_0 .
    \end{split}
\end{equation}

Now we can obtain an upper bound for the backlog at time $t$. Summing up \eqref{eq_long}, \eqref{eq_short}, \eqref{eq_short0}, and \eqref{eq_rem}, we have the maximum backlog
\begingroup
\allowdisplaybreaks
\begin{align*}
    \BW{0}{t}{t}
    &\,\le\, \sum_{i=0}^{N} \BW{10L_i \tau_i}{(10L_i + \ell_i - 1) \tau_i}{t} \\
        &\quad\ +\, \sum_{i=0}^{N} \BW{(10L_i + \ell_i - 1) \tau_i}{(10L_i + \ell_i) \tau_i}{t} \\
        &\quad\ +\, \BW{t-\epsilon}{t}{t} \\
    &\,\le\, \sum_{i=0}^{N} 3 \tau_i/k_i
        \,+\, \sum_{i=1}^{N} 3 \tau_{i-1}/k_{i-1}
        \,+\, 2 \tau_0 \\
    &\,\le\, \sum_{i=0}^{\infty} 6 \tau_i/k_i
        \,+\, 2 \tau_0 \\
    &\,=\, 60 D \sum_{i=0}^{\infty} (2/5)^i
        \,+\, 20 D
    \,=\, 120 D \,\in\, O(D).
\end{align*}
\endgroup

\section{Localized MBP: Hardness}\label{sec:local-hard}

Recall from Section~\ref{sec:localizedMBP} that in the localized MBP the player wins if she catches the adversary (or otherwise keeps the adversary from reaching the target value) and the adversary wins if he reaches the target value somewhere.
In this section, we prove the following theorem:
\begin{theorem}\label{pspacehard}
    The localized MBP is PSPACE-hard, even for a target value of~$2$.
\end{theorem}

\begin{proof}
We present a reduction from Quantified 3SAT (Q3SAT), where the Boolean formula $F$, containing $m$ clauses $c_1, c_2, \dotsc, c_m$ and $n$ variables $x_1, x_2, \dotsc,\allowbreak x_n$, is in conjunctive normal form with 3 literals per clause; without loss of generality, we assume that $n$ is even. A Q3SAT instance $I_F$ asks for the truthfulness of the expression $\forall x_1\exists x_2\forall x_3\ldots \exists x_n: c_1 \wedge c_2\wedge \ldots\wedge c_m$. It is helpful to think of this as a game between the player and the adversary who take turns at setting the variables in ascending order of indices; the player tries to set the odd variables in a way that will keep $F$ from being true, while the adversary sets the even variables in a way that aims at $F$ ending up satisfied.

Now we construct an instance of the localized MBP by specifying the digraph $D=(V,A)$ on which it is played; for more details of a somewhat related construction, see Fekete et al.~\cite{fffs-tspc-04}. The initial vertices for player and adversary are $u_{-1}$ and $u_0$, respectively, and the player starts the game. We use the vertex set
$V = \{x_i,\ov{x}_i,u_i : 1\le i\le n\}\cup\{u_{-1},u_0\}
\cup\{c_j,\overline{c}_j,d_j : 1\le j\le m\}$
and the edge set
\begin{align*}
    A &=\{(x_i,u_i),(\ov{x}_i,u_i) : 1\le i\le n\}
    \cup\{(u_i,x_{i+2}), (u_i,\ov{x}_{i+2}) : -1\le i\le n-2\}\\
    &\cup\{(u_{n},c_j), (u_{n-1},\overline{c}_j), (\overline{c}_j,d_j): 1\le j\le m\}
    \cup\{(\overline{c}_j,c_k) : 1\le j\le m, k\neq j\}\\
    &\cup\{(c_j,x_i), (d_j,x_i) : {\rm iff\ } c_j\ {\rm contains\ } x_i,\
    1\le j\le m,\ 1\le i\le n\}\\
    &\cup\{(c_j,\ov{x}_i), (d_j,\ov{x}_i) : {\rm iff\ } c_j\ {\rm contains\ } \ov{x}_i,\
    1\le j\le m,\ 1\le i\le n\}.
\end{align*}
We single out the subset $V_C=\{x_{2i-1},\ov{x}_{2i-1} : 1\le i\le n/2\} \subseteq V$ of vertices that start with an initial load of one; all other vertices start with an initial load of zero. Note that $|V|=O(m+n)$, $\ |V_C|=O(n)$, $|A|=O(m^2+n)$, so the construction is clearly polynomial. The construction is illustrated in Figures \ref{fig:variable} and~\ref{fig:clause}.

\begin{figure}
   \centering\input{variable.pdftex_t}
   \caption{The variable gadget: The player chooses a truth setting for the odd variables by running from $u_{-1}$ to $u_{n-1}$, while the adversary chooses a truth setting for the even variables by running from $u_{0}$ to $u_n$. Note that initially, the odd-numbered vertices carry a load of 1 (indicated by circles), while all other vertices start out with a load of~0.}
   \label{fig:variable}
\end{figure}
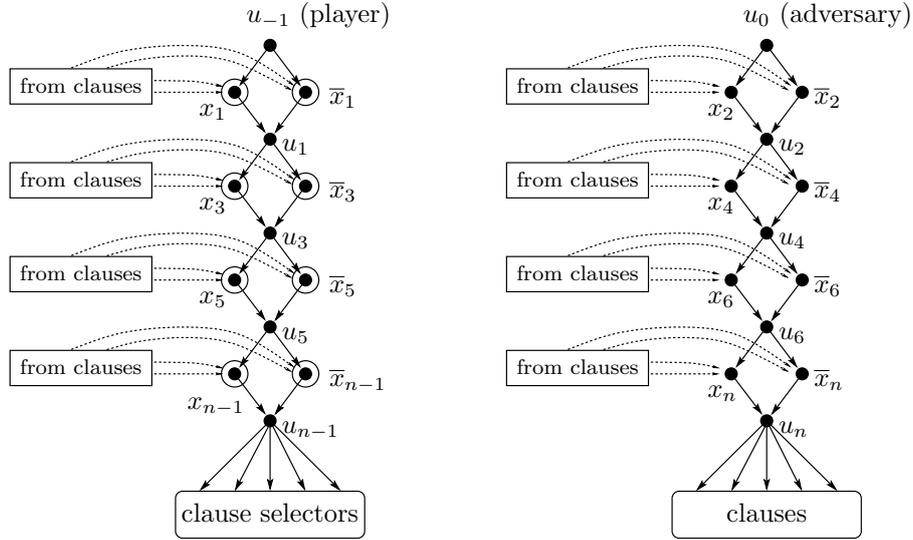

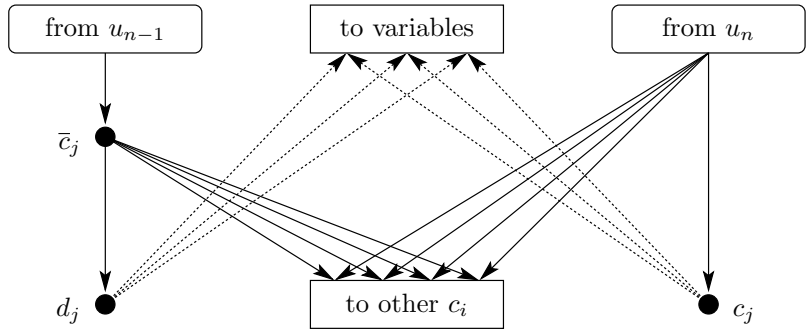
\begin{figure}
   \centering\input{clause.pdftex_t}
   \caption{A clause gadget: The player picks a clause by moving to a clause selector vertex $\overline{c}_j$, which is connected to all clause nodes $c_k$ for $k\neq j$. This forces the adversary to move to $c_j$ in order to avoid being caught prematurely. Then the player moves to $d_j$, catching the adversary after he moves to one of the three variable vertices corresponding to the clause $c_j$. The adversary wins if and only if that vertex already carries a load of 1, i.e., if the corresponding variable satisfies the clause.}
   \label{fig:clause}
\end{figure}

Now consider the game on $D$. For easier reference, we denote by \emph{diamonds} the subgraphs induced by $(u_{i-1},x_i,\ov{x}_i,u_i)$. The player and the adversary traverse the diamonds according to their chosen truth assignments in the given instance of Q3SAT, i.e., the adversary traverses $x_i$ if $x_i=1$ and otherwise traverses $\ov{x}_i$; analogously, the player traverses $\ov{x}_i$ if $x_i=1$ and otherwise traverses $x_i$; obviously, both participants are forced to move this way, implying a corresponding truth assignment. After the player (resp., adversary) arrives at $u_{n-1}$ (resp., $u_n$), for each $i$ exactly one of the vertices $x_i$ or $\ov{x}_i$ has a load of one.  We argue in the following that the adversary wins if and only if all clauses are satisfied.

After arriving at vertex $u_{n-1}$, the player selects a clause $c_j$ by moving to the corresponding clause selector vertex $\ov{c}_j$ (recall from Section~\ref{sec:localizedMBP} that the player starts first, and hence arrives to $u_{n-1}$ before the adversary arrives at $u_n$). This forces the adversary to move by $(u_n,c_j)$ in order to avoid being caught (note that there is no edge $(\overline{c}_j,c_j)$). As the player has no way of catching the adversary in her next move, the adversary wins if the clause vertex $c_j$ is adjacent to a variable vertex with load one, i.e., the corresponding variable setting satisfies the clause. On the other hand, the player can prevent the adversary from reaching a load of two if the clause is unsatisfied, by moving to vertex $d_j$, assuring herself of catching the adversary in her next move.

This shows that the player wins if and only if there is an unsatisfied clause.
\end{proof}

\section{Conclusions}

We have studied three versions of the MBP. We have shown that the localized MBP is hard to play optimally, while the geometric MBP is easy to play near-optimally (up to constant factors in the backlog). The hardness of the discrete MBP remains an open question.

For the geometric MBP, we have shown that the player has a strategy where the backlog does not depend on the number of cups, but only on the diameter of the cups set. This implies that the backlog scales linearly with the diameter of the area, and this is tight. An interesting open question is the scalability in the \emph{number of players}. If we have four players instead of one, we can divide the area into four parts and assign each player to one of the parts; this effectively halves the diameter and thus halves the backlog. It remains to be investigated whether we can exploit multiple players in a more efficient manner.

\section*{Acknowledgements}

We are very grateful to an anonymous referee who provided numerous suggestions that improved the presentation.
This article combines results from two preliminary conference articles~\cite{bfk-mbp-07,ps-obp-08}. We gratefully acknowledge Gerhard Woeginger for discussions that lead to the formulation of this problem. We thank Estie Arkin, Leonidas Guibas, Patrik Flor\'{e}en and Petteri Kaski for many helpful discussions.

MAB was supported in part by NSF Grants CCF-0621439/0621425, CCF-0540897/\allowbreak 05414009, CCF-0634793/0632838, CNS-0627645, CCF 1114809, CCF 1217708, IIS 1247726, IIS 1251137, CNS 1408695, and CCF 1439084.
AK was supported by DFG Grants FE407/9-1 and FE407/9-2.
VL was supported in part by NSF Grants CCR-0329910, Department of Commerce TOP 39-60-04003, Department of Energy DE-FC26-06NT42853, and the Wright Center for Sensor Systems Engineering.
JSBM was supported in part by the U.S.-Israel Binational Science Foundation (2000160, 2010074), NASA (NAG2-1620), NSF (CCF-0528209, ACI-0328930, CCF-0431030, CCF-1018388, CCF-1540890), and Metron Aviation.
JS was supported in part by the Academy of Finland, Grants 116547 and 118653 (ALGODAN), and by Helsinki Graduate School in Computer Science and Engineering (Hecse).

\bibliographystyle{abbrv}
\bibliography{cups}

\end{document}

%% file: decomp.pdftex_t
\begin{picture}(0,0)%
\includegraphics{decomp.pdftex}%
\end{picture}%
\setlength{\unitlength}{4144sp}%
\begingroup\makeatletter\ifx\SetFigFont\undefined%
\gdef\SetFigFont#1#2#3#4#5{%
  \reset@font\fontsize{#1}{#2pt}%
  \fontfamily{#3}\fontseries{#4}\fontshape{#5}%
  \selectfont}%
\fi\endgroup%
\begin{picture}(5402,8041)(-2475,730)
\put(-44,8354){\makebox(0,0)[rb]{\smash{{\SetFigFont{10}{12.0}{\familydefault}{\mddefault}{\updefault}{\color[rgb]{0,0,0}$\tau_0$}%
}}}}
\put(-44,7904){\makebox(0,0)[rb]{\smash{{\SetFigFont{10}{12.0}{\familydefault}{\mddefault}{\updefault}{\color[rgb]{0,0,0}$\tau_0$}%
}}}}
\put(-44,7544){\makebox(0,0)[rb]{\smash{{\SetFigFont{10}{12.0}{\familydefault}{\mddefault}{\updefault}{\color[rgb]{0,0,0}$\tau_1$}%
}}}}
\put(-44,6869){\makebox(0,0)[rb]{\smash{{\SetFigFont{10}{12.0}{\familydefault}{\mddefault}{\updefault}{\color[rgb]{0,0,0}$\tau_1$}%
}}}}
\put(-44,6194){\makebox(0,0)[rb]{\smash{{\SetFigFont{10}{12.0}{\familydefault}{\mddefault}{\updefault}{\color[rgb]{0,0,0}$\tau_2$}%
}}}}
\put(-44,5069){\makebox(0,0)[rb]{\smash{{\SetFigFont{10}{12.0}{\familydefault}{\mddefault}{\updefault}{\color[rgb]{0,0,0}$\tau_2$}%
}}}}
\put(-44,1694){\makebox(0,0)[rb]{\smash{{\SetFigFont{10}{12.0}{\familydefault}{\mddefault}{\updefault}{\color[rgb]{0,0,0}$\tau_3$}%
}}}}
\put(-44,3719){\makebox(0,0)[rb]{\smash{{\SetFigFont{10}{12.0}{\familydefault}{\mddefault}{\updefault}{\color[rgb]{0,0,0}$\tau_3$}%
}}}}
\put(1351,3719){\makebox(0,0)[lb]{\smash{{\SetFigFont{10}{12.0}{\familydefault}{\mddefault}{\updefault}{\color[rgb]{0,0,0}short period $3$}%
}}}}
\put(1351,6194){\makebox(0,0)[lb]{\smash{{\SetFigFont{10}{12.0}{\familydefault}{\mddefault}{\updefault}{\color[rgb]{0,0,0}short period $2$}%
}}}}
\put(1351,7544){\makebox(0,0)[lb]{\smash{{\SetFigFont{10}{12.0}{\familydefault}{\mddefault}{\updefault}{\color[rgb]{0,0,0}short period $1$}%
}}}}
\put(1351,8354){\makebox(0,0)[lb]{\smash{{\SetFigFont{10}{12.0}{\familydefault}{\mddefault}{\updefault}{\color[rgb]{0,0,0}short period $0$}%
}}}}
\put(-494,2594){\makebox(0,0)[rb]{\smash{{\SetFigFont{10}{12.0}{\familydefault}{\mddefault}{\updefault}{\color[rgb]{0,0,0}$(10 L_3 + 1) \tau_3$}%
}}}}
\put(-494,2819){\makebox(0,0)[rb]{\smash{{\SetFigFont{10}{12.0}{\familydefault}{\mddefault}{\updefault}{\color[rgb]{0,0,0}$(10 L_3 + \ell_3 - 1) \tau_3$}%
}}}}
\put(-494,4619){\makebox(0,0)[rb]{\smash{{\SetFigFont{10}{12.0}{\familydefault}{\mddefault}{\updefault}{\color[rgb]{0,0,0}$(10 L_3 + \ell_3) \tau_3 = 10 L_2 \tau_2$}%
}}}}
\put(-494,794){\makebox(0,0)[rb]{\smash{{\SetFigFont{10}{12.0}{\familydefault}{\mddefault}{\updefault}{\color[rgb]{0,0,0}$10 L_3 \tau_3 = 0$}%
}}}}
\put(-494,5744){\makebox(0,0)[rb]{\smash{{\SetFigFont{10}{12.0}{\familydefault}{\mddefault}{\updefault}{\color[rgb]{0,0,0}$(10 L_2 + \ell_2 - 1) \tau_2$}%
}}}}
\put(-494,5519){\makebox(0,0)[rb]{\smash{{\SetFigFont{10}{12.0}{\familydefault}{\mddefault}{\updefault}{\color[rgb]{0,0,0}$(10 L_2 + 1) \tau_2$}%
}}}}
\put(-494,6644){\makebox(0,0)[rb]{\smash{{\SetFigFont{10}{12.0}{\familydefault}{\mddefault}{\updefault}{\color[rgb]{0,0,0}$(10 L_2 + \ell_2) \tau_2 = 10 L_1 \tau_1$}%
}}}}
\put(-269,8624){\makebox(0,0)[rb]{\smash{{\SetFigFont{10}{12.0}{\familydefault}{\mddefault}{\updefault}{\color[rgb]{0,0,0}$t$}%
}}}}
\put(-494,8444){\makebox(0,0)[rb]{\smash{{\SetFigFont{10}{12.0}{\familydefault}{\mddefault}{\updefault}{\color[rgb]{0,0,0}$(10 L_0 + \ell_0) \tau_0 = T \tau_0 = t - \epsilon$}%
}}}}
\put(1621,5519){\makebox(0,0)[lb]{\smash{{\SetFigFont{10}{12.0}{\familydefault}{\mddefault}{\updefault}{\color[rgb]{0,0,0}$(10 L_3'+11) \tau_2$}%
}}}}
\put(1621,4619){\makebox(0,0)[lb]{\smash{{\SetFigFont{10}{12.0}{\familydefault}{\mddefault}{\updefault}{\color[rgb]{0,0,0}$(10 L_3'+10) \tau_2$}%
}}}}
\put(1621,2819){\makebox(0,0)[lb]{\smash{{\SetFigFont{10}{12.0}{\familydefault}{\mddefault}{\updefault}{\color[rgb]{0,0,0}$10 L_3' \tau_2$}%
}}}}
\put(1126,7994){\makebox(0,0)[lb]{\smash{{\SetFigFont{10}{12.0}{\familydefault}{\mddefault}{\updefault}{\color[rgb]{0,0,0}long period $0$}%
}}}}
\put(1126,7004){\makebox(0,0)[lb]{\smash{{\SetFigFont{10}{12.0}{\familydefault}{\mddefault}{\updefault}{\color[rgb]{0,0,0}long period $1$}%
}}}}
\put(1126,1829){\makebox(0,0)[lb]{\smash{{\SetFigFont{10}{12.0}{\familydefault}{\mddefault}{\updefault}{\color[rgb]{0,0,0}long period $3$}%
}}}}
\put(1126,5069){\makebox(0,0)[lb]{\smash{{\SetFigFont{10}{12.0}{\familydefault}{\mddefault}{\updefault}{\color[rgb]{0,0,0}long period $2$}%
}}}}
\put(  1,2729){\makebox(0,0)[b]{\smash{{\SetFigFont{10}{12.0}{\familydefault}{\mddefault}{\updefault}{\color[rgb]{0,0,0}$\ldots$}%
}}}}
\put(  1,5654){\makebox(0,0)[b]{\smash{{\SetFigFont{10}{12.0}{\familydefault}{\mddefault}{\updefault}{\color[rgb]{0,0,0}$\ldots$}%
}}}}
\put(  1,7229){\makebox(0,0)[b]{\smash{{\SetFigFont{10}{12.0}{\familydefault}{\mddefault}{\updefault}{\color[rgb]{0,0,0}$\ldots$}%
}}}}
\put(  1,8129){\makebox(0,0)[b]{\smash{{\SetFigFont{10}{12.0}{\familydefault}{\mddefault}{\updefault}{\color[rgb]{0,0,0}$\ldots$}%
}}}}
\end{picture}%

%% file: variable.pdftex_t
\begin{picture}(0,0)%
\includegraphics{variable.pdftex}%
\end{picture}%
\setlength{\unitlength}{3868sp}%
\begingroup\makeatletter\ifx\SetFigFont\undefined%
\gdef\SetFigFont#1#2#3#4#5{%
  \reset@font\fontsize{#1}{#2pt}%
  \fontfamily{#3}\fontseries{#4}\fontshape{#5}%
  \selectfont}%
\fi\endgroup%
\begin{picture}(5969,3463)(-117,-2623)
\put(4426,-1711){\makebox(0,0)[lb]{\smash{{\SetFigFont{10}{12.0}{\familydefault}{\mddefault}{\updefault}{\color[rgb]{0,0,0}$x_n$}%
}}}}
\put(4426,-511){\makebox(0,0)[lb]{\smash{{\SetFigFont{10}{12.0}{\familydefault}{\mddefault}{\updefault}{\color[rgb]{0,0,0}$x_4$}%
}}}}
\put(4426, 89){\makebox(0,0)[lb]{\smash{{\SetFigFont{10}{12.0}{\familydefault}{\mddefault}{\updefault}{\color[rgb]{0,0,0}$x_2$}%
}}}}
\put(4426,-1111){\makebox(0,0)[lb]{\smash{{\SetFigFont{10}{12.0}{\familydefault}{\mddefault}{\updefault}{\color[rgb]{0,0,0}$x_6$}%
}}}}
\put(5101,-436){\makebox(0,0)[lb]{\smash{{\SetFigFont{10}{12.0}{\familydefault}{\mddefault}{\updefault}{\color[rgb]{0,0,0}$\ov{x}_4$}%
}}}}
\put(4876,-736){\makebox(0,0)[lb]{\smash{{\SetFigFont{10}{12.0}{\familydefault}{\mddefault}{\updefault}{\color[rgb]{0,0,0}$u_4$}%
}}}}
\put(4876,-136){\makebox(0,0)[lb]{\smash{{\SetFigFont{10}{12.0}{\familydefault}{\mddefault}{\updefault}{\color[rgb]{0,0,0}$u_2$}%
}}}}
\put(5101,164){\makebox(0,0)[lb]{\smash{{\SetFigFont{10}{12.0}{\familydefault}{\mddefault}{\updefault}{\color[rgb]{0,0,0}$\ov{x}_2$}%
}}}}
\put(5101,-1036){\makebox(0,0)[lb]{\smash{{\SetFigFont{10}{12.0}{\familydefault}{\mddefault}{\updefault}{\color[rgb]{0,0,0}$\ov{x}_6$}%
}}}}
\put(4876,-1336){\makebox(0,0)[lb]{\smash{{\SetFigFont{10}{12.0}{\familydefault}{\mddefault}{\updefault}{\color[rgb]{0,0,0}$u_6$}%
}}}}
\put(5101,-1636){\makebox(0,0)[lb]{\smash{{\SetFigFont{10}{12.0}{\familydefault}{\mddefault}{\updefault}{\color[rgb]{0,0,0}$\ov{x}_n$}%
}}}}
\put(4876,-1936){\makebox(0,0)[lb]{\smash{{\SetFigFont{10}{12.0}{\familydefault}{\mddefault}{\updefault}{\color[rgb]{0,0,0}$u_n$}%
}}}}
\put(4651,689){\makebox(0,0)[lb]{\smash{{\SetFigFont{10}{12.0}{\familydefault}{\mddefault}{\updefault}{\color[rgb]{0,0,0}$u_0$ (adversary)}%
}}}}
\put(1726,-1936){\makebox(0,0)[lb]{\smash{{\SetFigFont{10}{12.0}{\familydefault}{\mddefault}{\updefault}{\color[rgb]{0,0,0}$u_{n-1}$}%
}}}}
\put(1726,-136){\makebox(0,0)[lb]{\smash{{\SetFigFont{10}{12.0}{\familydefault}{\mddefault}{\updefault}{\color[rgb]{0,0,0}$u_1$}%
}}}}
\put(1726,-736){\makebox(0,0)[lb]{\smash{{\SetFigFont{10}{12.0}{\familydefault}{\mddefault}{\updefault}{\color[rgb]{0,0,0}$u_3$}%
}}}}
\put(1726,-1336){\makebox(0,0)[lb]{\smash{{\SetFigFont{10}{12.0}{\familydefault}{\mddefault}{\updefault}{\color[rgb]{0,0,0}$u_5$}%
}}}}
\put(2026,164){\makebox(0,0)[lb]{\smash{{\SetFigFont{10}{12.0}{\familydefault}{\mddefault}{\updefault}{\color[rgb]{0,0,0}$\ov{x}_1$}%
}}}}
\put(2026,-436){\makebox(0,0)[lb]{\smash{{\SetFigFont{10}{12.0}{\familydefault}{\mddefault}{\updefault}{\color[rgb]{0,0,0}$\ov{x}_3$}%
}}}}
\put(2026,-1036){\makebox(0,0)[lb]{\smash{{\SetFigFont{10}{12.0}{\familydefault}{\mddefault}{\updefault}{\color[rgb]{0,0,0}$\ov{x}_5$}%
}}}}
\put(2026,-1636){\makebox(0,0)[lb]{\smash{{\SetFigFont{10}{12.0}{\familydefault}{\mddefault}{\updefault}{\color[rgb]{0,0,0}$\ov{x}_{n-1}$}%
}}}}
\put(1201,-1111){\makebox(0,0)[lb]{\smash{{\SetFigFont{10}{12.0}{\familydefault}{\mddefault}{\updefault}{\color[rgb]{0,0,0}$x_5$}%
}}}}
\put(1201,-511){\makebox(0,0)[lb]{\smash{{\SetFigFont{10}{12.0}{\familydefault}{\mddefault}{\updefault}{\color[rgb]{0,0,0}$x_3$}%
}}}}
\put(1201, 89){\makebox(0,0)[lb]{\smash{{\SetFigFont{10}{12.0}{\familydefault}{\mddefault}{\updefault}{\color[rgb]{0,0,0}$x_1$}%
}}}}
\put(1126,-1786){\makebox(0,0)[lb]{\smash{{\SetFigFont{10}{12.0}{\familydefault}{\mddefault}{\updefault}{\color[rgb]{0,0,0}$x_{n-1}$}%
}}}}
\put(1501,689){\makebox(0,0)[lb]{\smash{{\SetFigFont{10}{12.0}{\familydefault}{\mddefault}{\updefault}{\color[rgb]{0,0,0}$u_{-1}$ (player)}%
}}}}
\put(451,-1561){\makebox(0,0)[b]{\smash{{\SetFigFont{8}{9.6}{\familydefault}{\mddefault}{\updefault}{\color[rgb]{0,0,0}from clauses}%
}}}}
\put(451,-961){\makebox(0,0)[b]{\smash{{\SetFigFont{8}{9.6}{\familydefault}{\mddefault}{\updefault}{\color[rgb]{0,0,0}from clauses}%
}}}}
\put(451,-361){\makebox(0,0)[b]{\smash{{\SetFigFont{8}{9.6}{\familydefault}{\mddefault}{\updefault}{\color[rgb]{0,0,0}from clauses}%
}}}}
\put(451,239){\makebox(0,0)[b]{\smash{{\SetFigFont{8}{9.6}{\familydefault}{\mddefault}{\updefault}{\color[rgb]{0,0,0}from clauses}%
}}}}
\put(3601,-1561){\makebox(0,0)[b]{\smash{{\SetFigFont{8}{9.6}{\familydefault}{\mddefault}{\updefault}{\color[rgb]{0,0,0}from clauses}%
}}}}
\put(3601,-961){\makebox(0,0)[b]{\smash{{\SetFigFont{8}{9.6}{\familydefault}{\mddefault}{\updefault}{\color[rgb]{0,0,0}from clauses}%
}}}}
\put(3601,-361){\makebox(0,0)[b]{\smash{{\SetFigFont{8}{9.6}{\familydefault}{\mddefault}{\updefault}{\color[rgb]{0,0,0}from clauses}%
}}}}
\put(3601,239){\makebox(0,0)[b]{\smash{{\SetFigFont{8}{9.6}{\familydefault}{\mddefault}{\updefault}{\color[rgb]{0,0,0}from clauses}%
}}}}
\put(4801,-2506){\makebox(0,0)[b]{\smash{{\SetFigFont{10}{12.0}{\familydefault}{\mddefault}{\updefault}{\color[rgb]{0,0,0}clauses}%
}}}}
\put(1651,-2506){\makebox(0,0)[b]{\smash{{\SetFigFont{10}{12.0}{\familydefault}{\mddefault}{\updefault}{\color[rgb]{0,0,0}clause selectors}%
}}}}
\end{picture}%

%% file: clause.pdftex_t
\begin{picture}(0,0)%
\includegraphics{clause.pdftex}%
\end{picture}%
\setlength{\unitlength}{3947sp}%
\begingroup\makeatletter\ifx\SetFigFont\undefined%
\gdef\SetFigFont#1#2#3#4#5{%
  \reset@font\fontsize{#1}{#2pt}%
  \fontfamily{#3}\fontseries{#4}\fontshape{#5}%
  \selectfont}%
\fi\endgroup%
\begin{picture}(4974,2113)(1039,-1198)
\put(5551,-1111){\makebox(0,0)[lb]{\smash{{\SetFigFont{10}{12.0}{\familydefault}{\mddefault}{\updefault}{\color[rgb]{0,0,0}$c_j$}%
}}}}
\put(1501,-61){\makebox(0,0)[rb]{\smash{{\SetFigFont{10}{12.0}{\familydefault}{\mddefault}{\updefault}{\color[rgb]{0,0,0}$\ov{c}_j$}%
}}}}
\put(1501,-1111){\makebox(0,0)[rb]{\smash{{\SetFigFont{10}{12.0}{\familydefault}{\mddefault}{\updefault}{\color[rgb]{0,0,0}$d_j$}%
}}}}
\put(5401,644){\makebox(0,0)[b]{\smash{{\SetFigFont{10}{12.0}{\familydefault}{\mddefault}{\updefault}{\color[rgb]{0,0,0}from $u_n$}%
}}}}
\put(3526,644){\makebox(0,0)[b]{\smash{{\SetFigFont{10}{12.0}{\familydefault}{\mddefault}{\updefault}{\color[rgb]{0,0,0}to variables}%
}}}}
\put(3526,-1081){\makebox(0,0)[b]{\smash{{\SetFigFont{10}{12.0}{\familydefault}{\mddefault}{\updefault}{\color[rgb]{0,0,0}to other $c_i$}%
}}}}
\put(1651,644){\makebox(0,0)[b]{\smash{{\SetFigFont{10}{12.0}{\familydefault}{\mddefault}{\updefault}{\color[rgb]{0,0,0}from $u_{n-1}$}%
}}}}
\end{picture}%